\documentclass[journal,article,submit,pdftex,moreauthors]{Definitions/mdpi} 

\firstpage{1} 
\pubvolume{1}
\issuenum{1}
\articlenumber{0}
\pubyear{2024}
\copyrightyear{2024}
\datereceived{ } 
\daterevised{ } % Comment out if no revised date
\dateaccepted{ } 
\datepublished{ } 
\hreflink{https://doi.org/} % If needed use \linebreak
\newcommand{\Du}{\nabla_{\!\boldsymbol{u}}} 
\newcommand{\Dtu}{\nabla_{\!\boldsymbol{\tilde{u}}}} 
\newcommand{\Db}{\nabla_{\!\boldsymbol{b}}}
\newcommand{\idop}{\mathbb{I}}
\newcommand{\dd}{\mathrm{d}}

\usepackage{amsmath,amssymb,amsfonts,amsthm}

\newtheorem{thm}{Theorem}[section]
\newtheorem{prop}[thm]{Proposition}

\theoremstyle{definition}
\newtheorem{defi}[thm]{Definition}

\Title{Foliation-generating observers under Lorentz transformations}

\TitleCitation{Foliation-generating observers under Lorentz transformations}

\Author{Daniel Blixt $^{1,\ddagger}$\orcidA{}, Alejandro Jiménez Cano $^{2,3\ddagger}$\orcidB{}, and Aneta Wojnar $^{4,5\ddagger}$\orcidC{}}

\AuthorNames{Daniel Blixt, Alejandro {Jim{\'e}nez Cano}, and Aneta Wojnar}

\AuthorCitation{Blixt, D.; Jim{\'e}nez Cano, A.; Wojnar, A.}
\address{%
$^{1}$ \quad Scuola Superiore Meridionale, Largo S. Marcellino 10, I-80138 Napoli, Italy; d.blixt@ssmeridionale.it\\
$^{2}$ \quad Escuela T\'ecnica Superior de Ingenier\'ia de Montes, Forestal y del Medio Natural, Universidad Politécnica de Madrid, 28040 Madrid, Spain; alejandro.jimenez.cano@upm.es\\
$^{3}$ \quad Laboratory of Theoretical Physics, Institute of Physics, University of Tartu,W. Ostwaldi 1, Tartu 50411, Estonia\\
$^{4}$ \quad Department of Theoretical Physics \& IPARCOS, Complutense University of Madrid, E-28040, 
Madrid, Spain; awojnar@ucm.es\\
$^{5}$ \quad Institute of Theoretical physics, University of Wroc\l aw, pl. Maxa Borna 9, 50-206 Wroc\l aw, Poland
}

\corres{Correspondence: d.blixt@ssmeridionale.it}

\secondnote{These authors contributed equally to this work.}
\abstract{
In this work, we revise the concept of foliation and related aspects that are crucial when formulating the Hamiltonian evolution for various theories beyond General Relativity. In particular, we show the relation between the kinematic characteristics of timelike congruences (observers) and the existence of foliations orthogonal to them. We then explore how local Lorentz transformations acting on observers affect the existence of transversal foliations, provide examples, and discuss the implications of these results for the $3+1$ formulation of tetrad modified theories of gravity. }


\begin{document}

%%%%%%%%%%%%%%%%%%%%%%%%%%%%%%%%%%%%%%%%%%

\section{Introduction}

The problem of evolving in time the metric in General Relativity can be regarded as a Cauchy problem, where the evolution of the system is uniquely given by the initial positions and velocities of the fundamental fields. This accounts to specify the induced metric $\gamma_{ij}$ and $\partial_t \gamma_{ij}$ at a given instant of time, that is, to specify all the metric components and their first time derivatives at a given spatial slice $x^0=t=\text{constant}$. The evolution equations must be supplemented with the constraints fulfilled by the fields and their derivatives, the well-known Hamiltonian and momenta constraints of General Relativity. Moreover, the study of the time evolution also requires breaking the spacetime covariance through the $3+1$ split of spacetime. Then the so-called ADM equations fully describe the nonlinear behavior of General Relativity, although it has been proven that they lack the required strong hyperbolicity in order to sustain computer simulations of them. For this purpose, it is resorted to the BSSNOK formalism \cite{Baumgarte:2010ndz, Baumgarte:2021skc}. 

In recent years, there has been a growing interest in modified theories of gravity \cite{CANTATA:2021ktz}, many of which have been formulated in the so-called metric-affine framework, which incorporates a dynamical connection in the field content. This framework includes different kinds of theories characterized by different constraints \textit{a priori} on some of their fields, for instance generic metric-affine theories \cite{HEHL19951} (no constraints), Poincaré gauge theories \cite{Baekler:2011jt} (metric-compatible connection), and also teleparallel theories \cite{Bahamonde:2021gfp, BeltranJimenez:2019odq, BeltranJimenez:2019esp} (zero curvature). In these theories, a frame (the {\it tetrad}) often appears as a fundamental variable instead of the metric, and local Lorentz symmetries often enter the game, allowing to remove some of these additional degrees of freedom. In fact, some theories might present extra copies of the Lorentz symmetry realized on different fields, but all of them affect the tetrad configuration (see \cite{Blixt:2022rpl}). These symmetry transformations connect different tetrad solutions, which will correspond to different observers, which are the timelike vectors of the frame.

The chosen observer plays a relevant role in the $3+1$ decomposition since the parameters of the associated congruence define the time function used to evolve the initial Cauchy hypersurface. It can be checked that not all timelike congruences have an associated spacelike foliation, as we will revise. The point is that there are many works in the literature for which a $3+1$ decomposition of spacetime has been adopted without checking if the solutions for the tetrads fulfill the conditions for the existence of a foliation. Even in the context of General Relativity, this condition may not always be satisfied \cite{poisson2004relativist, awmgr, Borowiec:2013kgx}. Therefore, it is crucial to carefully evaluate whether the procedures used for the $3+1$ decomposition remain valid in modified gravity theories. So far, a $3+1$ decomposition has been adopted for the Hamiltonian analysis in metric teleparallel gravity (see \cite{Blixt:2020ekl} for a review) and for considering the evolution equations \cite{Capozziello:2021pcg} and Hamilton's equations \cite{Pati:2022nwi} for the (Metric) Teleparallel Equivalent of General Relativity. 
The same has been done for the Hamiltonian analysis of Poincaré gauge theories of gravity \cite{Blagojevic:2000xd} and Einstein-Cartan theory \cite{Kiriushcheva:2010sss, Kiriushcheva:2009tg}.

The aim of this work is to provide a comprehensive clarification of the mathematical concept of foliation within the framework of a $3+1$ decomposition and the implications of performing Lorentz transformations on a given observer with an associated foliation. 

This article is structured as follows: In Sec.~\ref{sec:Math} the mathematical preliminaries are introduced in two parts: in Sec.~\ref{sec:Tetrad} we give the notion of tetrads and tetrad formulation of General Relativity and other theories, and Sec.~\ref{sec:Foliations} introduces the formalism for foliations and the kinematic characteristics of a congruence of observers. We also present the conditions required for the existence of a foliation associated to a given observer. Then, Sec.~\ref{sec:Condition} explores how these conditions behave under Lorentz transformations. In Sec.~\ref{sec:Examples} we provide some explicit examples of observers that do, and do not, satisfy the condition for the existence of foliation. Finally, Sec.~\ref{sec:Conclusions} is devoted to a summary and a brief discussion, as well as the conclusions of this work. At the end, we also included the Appendix~\ref{app:distrib}, which contains more details on the relation between distributions and foliations.

~

\noindent{\bf Conventions.}
We work in arbitrary spacetime dimension $n$, although the examples given are four-dimensional. Our notation is such that Greek letters $\mu,\nu,\rho,\ldots$ denote spacetime indices. The indices with respect to an orthonormal frame or tetrad (Lorentz indices) are denoted by the uppercase first letters of the Latin alphabet $A,B,C,\ldots$, running from $0$ to $n-1$, and the spatial directions of the tetrad are covered by the indices $i,j$. We choose the convention $(-,+,+,+)$ for the metric signature. $\Gamma$ and $\nabla$ denote the Levi-Civita connection and covariant derivative, and we will often abbreviate $\nabla_{\!\boldsymbol{X}} \equiv X^\mu \nabla_\mu$.

\section{Mathematical preliminaries}
\label{sec:Math}

In the following section, we will focus on introducing tetrad formulation and the geometric set-up required by the $3+1$ formalism. We will start with basic definitions and theorems for the reader's convenience: we will recall the notions of a tangent distribution, foliation, and almost-product structure, where the last one is a useful tool to study geometry and an integrability of distributions (that is, foliations). In the last part, we will relate some of those geometric structures with an observer's motion and how the $3+1$ decomposition depends on it. We will also provide the corresponding conclusions in the language of tetrad, which will be useful to study teleparallel theories of gravity. 

\subsection{Tetrad theories of gravity}
\label{sec:Tetrad}

In any metric manifold $(\mathcal{M}, g)$, we construct a field of tetrads (a basis of orthonormal vectors $\{e_A\}$) at each point of $\mathcal{M}$ with components $e_A{}^{\mu}$. If we denote the components of the dual cotetrad as $E^{A}{}_{\mu}$, the following relations hold:
\begin{equation}
    g_{\mu\nu} = \eta_{AB} E^A{}_{\mu} E^B{}_{\nu}, \qquad   \eta_{AB} = g_{\mu\nu} e_{A}{}^{\mu} e_{B}{}^{\nu}.
    \label{eq:fundrel}
\end{equation}
By definition, the tetrad and cotetrad components also satisfy orthonormality relations
\begin{equation}
E^A{}_{\mu}  e_B{}^{\mu} = \delta^{A}{}_{B}, \qquad E^A{}_{\mu}  e_A{}^{\nu} = \delta^{\nu}{}_{\mu}.
\label{eq:orth}
\end{equation}
We can transform Lorentz indices of a vector $V^{A}$ into spacetime indices $V^{\mu}$ with the components of the tetrad such that $V^{\mu} = e_{A}{}^{\mu} V^{A}$, and the opposite $V_{A} = e_{A}{}^{\mu} V_{\mu}$. Lorentz indices are raised and lowered with the Minkowski metric, while spacetime indices are raised and lowered with the spacetime metric.

Using the aforementioned relations, it is straightforward to reformulate any metric theory in terms of tetrads. For instance, after expressing the volume element as $ \sqrt{-g}=|\det{(E^A{}_\mu)}|$ and substituting the metrics appearing in the Ricci scalar by tetrads, the tetrad formulation of General Relativity can be obtained with no difference in physical predictions. Although the 10 components of $g_{\mu\nu}$ are promoted to 16 of $e_A{}^\mu$ in the tetrad formulation, a local Lorentz symmetry is introduced in the theory, and those additional components can be eliminated by choosing an appropriate gauge (see e.g. \cite{Ortin}).

\subsection{Foliations in physics}
\label{sec:Foliations}

\subsubsection{Introduction and basic examples}

Consider a spacetime $(\mathcal{M},g)$ of dimension $n$, i.e., a smooth manifold equipped with a Lorentzian metric. We denote $g_{\mu\nu}$ the component of the metric tensor with respect to some arbitrary coordinates. Before recalling the basic notions in a more formal way, let us get familiar with the notion of foliation and its relation to the choice of an observer. In an open set $U\subset\mathcal{M}$ (which can be the whole spacetime), any smooth normalized timelike vector field $u^\mu$  will be called an \emph{observer} on $U$. In our convention,  $g_{\mu\nu}u^\mu u^\nu=-1$. Notice that an observer can be understood as an arrow of time in the region it is defined. An arbitrary observer can always be written locally as 
\begin{equation}\label{constant}
    u^\mu = (1,0,\ldots,0).
\end{equation}
In other words, we can always find a set of coordinates in which $u^\mu\partial_\mu = \partial_t$ is the first vector of the associated basis, and the corresponding coordinate $t$ plays the role of time.  
When we are able to do this globally in $\mathcal{M}$, we can think about the spacetime manifold $ \mathcal{M}$ as a composition of spacelike hypersurfaces such that $t=\text{const}$ for each of them, and such that they are diffeomorphic to some $(n-1)$-dimensional manifold $\Sigma$. In other words, the entire spacetime admits a splitting of the form:
\begin{equation}\label{global}
  \mathcal  M=\mathcal{T}\times\Sigma,
\end{equation}
where the factor $\mathcal{T}$ of the product corresponds to the 1-dimensional integral curves of the observer field $u^\mu$. When this happens, we say that the spacetime is (globally) foliated by the hypersurfaces of constant $t$, all of them diffeomorphic to $\Sigma$. If, in addition, the spatial hypersurfaces are Cauchy, the spacetime is called globally hyperbolic. Notice that such a splitting might not be possible for some particular spacetimes and chosen observers. To see that, see an example in the subsection \ref{geoobs}. The simplest example of the global splitting is provided for the four-dimensional Minkowski spacetime by the constant vector field \eqref{constant}. The counterexample within the same spacetime is the Rindler observer (where $x$ and $t$ are Cartesian coordinates):
\begin{equation}\label{rindler}
u^\mu = \left(\frac{x}{x^2-t^2}, \frac{t}{x^2-t^2}, 0,..., 0\right),
\end{equation}
which foliates only a quarter of the spacetime, with the remainder hidden behind an event horizon from the perspective of this observer.

Now on, let us formalize a bit the above discussion. 

\subsubsection{Distributions of codimension one. Foliations}

Let $(\mathcal{M},g)$ be an $n$-dimensional pseudo-Riemannian manifold and $T\mathcal{M}$ denote its 
tangent bundle. A (regular) $k$-distribution $\mathcal{D}$ is a smooth map that assigns to any point $p\in\mathcal{M}$ a subspace (hyperplane) of $T_p\mathcal{M}$ of dimension $k$:
\begin{equation}
    \mathcal{D}\,: \,p\rightarrow\;\mathcal{D}_p\subset T_p\mathcal{M}\,.
\end{equation}
In other words, a $k$-distribution is a $k$-dimensional subbundle $\mathcal{D}$ of $T\mathcal{M}$. Locally on any open set $U\subset\mathcal{M}$, a $k$-distribution is generated by a set of
$k$ linearly independent vector fields on $U$, $\{X_i\}^k_{i=1}$. These vector fields, at each point $p\in U$, span the $k$-dimensional subspace $\mathcal{D}_p\subset T_p\mathcal{M}$, i.e., $\mathcal{D}_p=\mathrm{span}\{X_1(p),...,X_k(p)\}$.\footnote{
    In the language of fiber bundles, $X_i\in\Gamma(\mathcal{D})$, where $\Gamma$ denotes the space of sections.} 
The distribution $\mathcal{D}$ is said to be involutive when the Lie bracket of any two vector fields of the basis $\{X_i\}$ is also a vector field belonging to the distribution.

In a physical context, we are mainly interested in distributions that are \emph{completely integrable}. Let us introduce this concept. First, consider an embedded submanifold $\mathcal{N}\subset\mathcal{M}$ such that $T_p\mathcal{N}=\mathcal{D}_p$ in every
point $p\in\mathcal{N}$. Then $\mathcal{N}$ is called an integral manifold of the distribution $\mathcal{D}$. Then, a $k$-distribution $\mathcal{D}$ is completely integrable if for each point $p\in\mathcal{M}$ there exists an integral $k$-dimensional manifold $\mathcal{N}$ of $\mathcal{D}$ passing through $p$. Local Frobenius theorem \cite{lee2012manifolds} states that  involutive distributions are completely integrable. Notice that, trivially, this shows that every smooth $1$-dimensional distribution is integrable.

There is another result called the Global Frobenius theorem (see also \cite{lee2012manifolds}) that provides a relation between integrability and the existence of foliations. Accordingly, the collection of all maximal
connected integral manifolds of an involutive $k$-dimensional distribution $\mathcal{D}$ on a smooth manifold $\mathcal{M}$ forms a foliation of $\mathcal{M}$. In other words, a foliation is a collection of submanifolds $\mathcal{N}_i$ in a way that each of $\mathcal{N}_i$ goes on smoothly and into each other and does not intersect each other.
More rigorous and detailed discussion can be found in \cite{lee2012manifolds,awmgr,Borowiec:2013kgx}.  

Let us focus now on the case of distributions of codimension one, i.e., $\mathcal{D}$ is an $(n-1)$-distribution on the $n$-dimensional pseudo-Riemannian manifold $(\mathcal{M},g)$. We further assume that the hyperplanes of the distribution are everywhere non-null. Then, at every point $p\in\mathcal{M}$
\begin{equation}\label{decomp1}
T_p\mathcal{M}=\mathcal{D}_p\oplus_g\mathcal{D}_p^\perp. 
\end{equation}
In general, those distributions are not complementary. As $\mathrm{dim}\mathcal{D}_p=n-1$, we deduce that $\mathrm{dim}\mathcal{D}_p^\perp=1$, therefore it is integrable, as discussed before. If $\mathcal{D}$ satisfies the conditions of being involutive (and therefore, completely integrable by the local Frobenius theorem) and further of being a foliation, we may consider a particular example of the global foliation \cite{hawking2023large} as discussed before.

\subsubsection{Geometry of foliations}\label{geoobs}

Geometry of distributions and foliations can be easily studied with the use of Naveira's classification of the almost-product manifolds \cite{naveira1983classification} and the properties \cite{gil1983geometric} of the almost-product structure on $\mathcal{M}$ \cite{gray1967pseudo,yano1985structures}. A reader who is interested in the rigorous derivation and discussion is invited to study the mentioned references. We have also provided a nutshell discussion in Appendix \ref{app:distrib}. Since we will only focus on the physical meaning provided by this formalism, as presented in \cite{Borowiec:2013kgx}, in what follows, we will just recall the main results relevant for the further discussion.

\begin{defi}\label{defkinem}
Let $u^\mu$ be an observer $n$-field. 
The kinetic characteristics of the observer's motion through the spacetime are the vorticity tensor, the shear tensor, the expansion scalar, and the acceleration vector, respectively, given by
\begin{align}
\omega_{\mu\nu}&:= \nabla_{[\nu} u_{\mu]}+\dot{u}_{[\mu}u_{\nu]}\,,\\
\sigma_{\mu\nu}&:=\nabla_{(\nu}u_{\mu)}+\dot{u}_{(\mu}u_{\nu)}-\frac{1}{n-1}\Theta h_{\mu\nu}\,,\label{shear}\\
\Theta&:= \nabla_\rho u^\rho\,,\label{expan}\\
\dot{u}^\mu&:=u^\rho \nabla_\rho u^\mu\,,
\end{align}
where the brackets $[~]$ and $(~)$ denote antisymmetrization and symmetrization, respectively, while $h_{\mu\nu}:=g_{\mu\nu}+u_{\mu}u_{\nu}$ is the tensor projecting on a three-dimensional subspace.
\end{defi}

Let us notice that the acceleration vector, vorticity, and shear tensors are spacelike fields, since $u^{\mu}\dot{u}_{\mu}=u^{\mu}\omega_{\mu\nu}=u^{\mu}\sigma_{\mu\nu} =0.$
The above characteristics of the observer's motion are the irreducible components of the projected tensor $\nabla_\nu u_{\mu}$, which determines the rate of change in
the position of one point with respect to the other one in the fluid material \cite{Ehlers:1961xww,plebanski2024introduction}. As already discussed, we will focus on theories of gravity in which the matter is solely coupled to the metric, and the observer field represents a 4-velocity of fluid. Therefore, this decomposition is done with respect to the Levi-Civita connection compatible with the metric.

As it turns out, the above kinetic characteristics are related to the geometry of distributions and foliations. For the reader's convenience, let us recall the following theorems \cite{awmgr,Borowiec:2013kgx}:

\begin{thm}
    Let $u^\mu$ be an observer $n$-field. The $1$-distribution (foliation) is totally geodesic if $\dot u^\mu=0$.
\end{thm}

\begin{thm}\label{th:condit}
    Let $u^\mu$ be an observer $n$-field. The $(n-1)$-distribution is
\begin{itemize}
    \item  minimal if $\Theta=0$
    \item  umbilical if $\sigma_{\mu\nu}=0$
    \item  geodesic if $\sigma_{\mu\nu}=0$ and $\Theta=0$
    \item  a foliation when $\omega_{\mu\nu}=0$. The foliation is minimal, totally umbilical or totally geodesic if $\Theta=0$,  $\sigma_{\mu\nu}=0$, and  $\sigma=\Theta=0$, respectively.
\end{itemize}
\end{thm}
The definitions of the above notions related to the geometry of distribution and foliations are recalled in the Appendix~\ref{app:distrib}. The more detailed discussion and proof are provided in \cite{awmgr,Borowiec:2013kgx}. Note that $\omega_{\mu\nu}=0$ is equivalent to vanishing of the Nijenhuis tensor \cite{gray1967pseudo}.

Let us come back now to the Minkowski spacetime. It can be shown \cite{Borowiec:2013kgx} that the choice of the observers \eqref{constant} and \eqref{rindler} provide that the $(n-1)$-distribution is a foliation. A counterexample in four dimensions is the rotating observer in the $(x, y)$ plane
\begin{equation}\label{rot}
       u^\mu = \left(2,\frac{-y}{\sqrt{x^2+y^2}},\frac{x}{\sqrt{x^2+y^2}},0\right).
\end{equation}

\subsubsection{Relations between different observers}
For the purpose of the further part, let us consider two observers $u^\mu$ and $\tilde{u}^\mu$ (i.e., both of them normalized $u^\mu u_\mu = \tilde{u}^\mu \tilde{u}_\mu=-1$) on the same open set $U\subset M$. In general, 
\begin{equation}\label{transf}
    \tilde{u}^\mu = f u^\mu + b^\mu\,,
\end{equation}
for a certain real function $f$ and a vector field $b^\mu$ satisfying $u^\mu b_\mu=0$ and $b^\mu b_\mu = f^2 -1$. Notice that the orthogonality between $u^\mu$ and $b^\mu$ implies that $b_\mu$ is spacelike and that $b^\mu h_{\mu\nu} = b_\nu$. The metric can be split according to each of these observers, leading to two different projector tensors:
\begin{equation}
    g_{\mu\nu} = - u_\mu u_\nu + h_{\mu\nu} = - \tilde{u}_\mu \tilde{u}_\nu + \tilde{h}_{\mu\nu}\,.
\end{equation}
 Notice that this implies:
\begin{equation}
    h_{\mu\nu} = (1-f^2)u_{\mu}u_\nu-2 f u_{(\mu} b_{\nu)} - b_\mu b_\nu - \tilde{h}_{\mu\nu}\,,
\end{equation}
leading to:
\begin{equation}
    u^\mu \tilde{h}_{\mu\nu} = (f^2-1)u_\nu + f b_\nu = - \frac{1}{f} b^\mu \tilde{h}_{\mu\nu} \,,
\end{equation}
consistent with $\tilde{u}^\mu \tilde{h}_{\mu\nu}=0$. 

Therefore, we can consider the decomposition of both covariant derivatives:
\begin{align}
    \nabla_\rho u_\mu &= \omega_{\mu\rho}+\sigma_{\mu\rho}+\frac{1}{n-1} h_{\mu\rho} \Theta -\dot{u}_\mu u_\rho \,,\label{eq:decom_Du}\\
    \nabla_\rho \tilde{u}_\mu &= \tilde{\omega}_{\mu\rho}+\tilde{\sigma}_{\mu\rho}+\frac{1}{n-1} \tilde{h}_{\mu\rho} \tilde{\Theta} -\tilde{\dot{u}}_\mu \tilde{u}_\rho \,,
\end{align}
where $\tilde{\dot{u}}^\mu:=\tilde{u}^\rho\nabla_\rho \tilde{u}^\mu$. The characteristic quantities associated to both observers are related by:
\begin{align}
    \tilde{\dot{u}}^\mu &= f^2 \dot{u}^\mu +f u^\mu  \Du f + f \Du b^\mu + \Db(f u^\mu + b^\mu)\,,\\[5pt]
    \tilde{\Theta} &= f \Theta + \Du f + \nabla_\rho b^\rho\,,\\[5pt]
    \tilde{\omega}_{\mu\nu}
    &= f \omega_{\mu\nu} + A_{[\mu\nu]}\,,\label{eq:rel_rot2} \\[5pt]
    \tilde{\sigma}_{\mu\nu}&= f \sigma_{\mu\nu} + A_{(\mu\nu)} - \frac{1}{n-1} B_{\mu\nu} \,,
\end{align}
where
\begin{align}
    A_{\mu\nu} &:= u_{\mu} \nabla_{\nu}f + \nabla_{\nu} b_{\mu} + f \dot{u}_{\mu} \Big(f \tilde{u}_{\nu} - u_{\nu}\Big) +   u_{\mu} b_{\nu} \Dtu f\nonumber\\
    &\qquad
    +\Big(\Dtu b_{\mu}+f\chi_{\mu}\Big)\tilde{u}_{\nu} + f u_\mu u_\nu \Dtu f \,,\label{eq:defA}\\
    B_{\mu\nu} &:= \tilde{h}_{\mu\nu}\tilde{\Theta} - f h_{\mu\nu}\Theta \\
    &=  \big(\tilde{u}_\mu\tilde{u}_\nu - u_\mu u_\nu + h_{\mu\nu}\big) \big(\Du f + \nabla_\rho b^\rho\big) +  \big(\tilde{u}_\mu\tilde{u}_\nu - u_\mu u_\nu\big)f\Theta \,,\label{eq:defB}
\end{align}
with
\begin{equation}
    \chi_\mu := b^\rho \nabla_\rho u_\mu =  \omega_{\mu\rho}b^\rho+\sigma_{\mu\rho}b^\rho+\frac{1}{n-1}  b_{\mu}\Theta \,.
\end{equation}
In addition, although the right-hand sides of \eqref{eq:defA} and \eqref{eq:defB} are intended to be seen as functions of the basic variables $\{u^\mu, f, b^\mu\}$ and the derivatives of $\{f, b^\mu\}$, we kept some $\tilde{u}^\mu$ (which can be eliminated with \eqref{transf}) to shorten the equation.

If we focus on \eqref{eq:rel_rot2} and expand $\chi_\mu$ we get:
\begin{align}
    \tilde{\omega}_{\mu\nu} &= f \omega_{\mu\nu} + u_{[\mu} \nabla_{\nu]}f + \nabla_{[\nu} b_{\mu]}+ f \dot{u}_{[\mu} \Big(f \tilde{u}_{\nu]} - u_{\nu]}\Big)\nonumber\\
    &\quad +   u_{[\mu} b_{\nu]} \left(\Dtu f-\frac{1}{n-1}\Theta f^2\right)+\Big[\Dtu b_{[\mu}+fb^\rho (\sigma_{\rho[\mu}-\omega_{\rho[\mu})\Big]\tilde{u}_{\nu]}\,.\label{eq:rel_rot3}
\end{align}
If both observers $u^\mu$ and $\tilde{u}^\mu$ define (not necessarily identical) foliations, we can drop both vorticities and find
\begin{align}
    0 &= u_{[\mu} \nabla_{\nu]}f + \nabla_{[\nu} b_{\mu]}+ f \dot{u}_{[\mu} \Big(f \tilde{u}_{\nu]} - u_{\nu]}\Big) \nonumber\\
    &\quad+   u_{[\mu} b_{\nu]} \left(\Dtu f-\frac{1}{n-1}\Theta f^2\right)
    +\Big[\Dtu b_{[\mu}+fb^\rho \sigma_{\rho[\mu}\Big]\tilde{u}_{\nu]}\,,\label{eq:rel_2fol}
\end{align}
For a fixed initial observer $u^\mu$, this is a system of differential equations for $\{f,b^\mu\}$ or, equivalently, $\{\tilde{u}^\mu\}$ (i.e., $n-1$ variables, due to the normalization condition).

\section{Foliation condition under Lorentz transformations} \label{sec:Condition}

The tetrad basis can always be oriented to the observer 4-field $u^{\mu}$ as 
\begin{align}
    u^{\mu}=e_{0}{}^{\mu}, \qquad u_{\mu} = g_{\mu\nu}e_{0}{}^{\nu}.
\end{align}
Another observer can be related by a Lorentz transformation of $u^\mu$
\begin{align}
\label{eq:LorentzObs}
    \tilde{u}^\mu=\tilde{e}_{0}{}^\mu = \Lambda^A{}_{0}e_A{}^\mu,
\end{align}
which can be written as \eqref{transf} when $f=\Lambda^0{}_{0}$ and $b^\mu=\Lambda^{i}{}_{0}e_{i}{}^{\mu}$.

A general Lorentz matrix can be generated by boosts and rotations. A \textit{ pure Lorentz rotation} satisfy $\Lambda^A{}_{0}=\delta^A{}_{0}$ and $\Lambda^{0}{}_{A}=\delta^{0}{}_{A}$. It is easy to see that
the following statements are true:

\begin{prop}\label{prop_rotation}
    Assume $u^\mu$ admits foliation. Then the observer corresponding to a pure Lorentz rotation also admits foliation. Moreover, the geometry of the resulting foliation is unchanged with respect to the previous one.
\end{prop} 
\begin{proof}
    A pure Lorentz rotation matrix ${}^\mathrm{R}\Lambda^A{}_B$ is given by
\begin{align}
    {}^\mathrm{R}\Lambda^A{}_B =\begin{pmatrix}
    1 & 0 & 0 & 0 \\
    0 & & & & \\
    0 & & {}^\mathrm{R}\Lambda^{i}{}_{j} & \\
    0 & & & & 
    \end{pmatrix}.
\end{align}
From eq \eqref{eq:LorentzObs} it follows that
\begin{align}\label{trivial}
    {}^R\tilde{u}^\mu={}^{\mathrm{R}}\Lambda^{A}{}_{0}e_A{}^\mu{}=e_{0}{}^\mu=u^\mu.
\end{align}
Since $u^\mu$ admits foliation and ${}^{\mathrm{R}}\tilde{u}^\mu=u^\mu$, we conclude that ${}^{\mathrm{R}}\tilde{u}^\mu$ also admits foliation with unchanged geometry. 
\end{proof}

However, Lorentz boosts generically give $\Lambda^A{}_{0}\neq \delta^A{}_{0}$. Thus, there is a condition required to understand if boosted observers admit foliation. From the discussion in the previous section, we have found that for the transformation of the type \eqref{transf}, it requires vanishing of the transformed vorticity tensor \eqref{eq:rel_rot3}.
Since we are interested in Lorentz transformations of the observer field which preserve the foliated structure of the spacetime manifold, let us propose the following:
\begin{thm}\label{th:cond_foliation}
    Let the function $f$ be $f = \Lambda^0{}_{0}$ while the vector $b^\mu$ is given by $b^\mu=\Lambda^{i}{}_{0}e_{i}{}^\mu$. Let the observer $u^\mu=e_{0}{}^\mu$ satisfy the condition for the foliation. Then, the observer $\tilde u^\mu=\Lambda^A{}_{0}e_A{}^\mu$ transformed by the Lorentz transformation will also admit foliations if the following equation is satisfied:
    \begin{align}\label{theo}
    0 &=  u_{[\mu} \nabla_{\nu]}f + \nabla_{[\nu} b_{\mu]} + f \dot{u}_{[\mu} \Big(f \tilde{u}_{\nu]} - u_{\nu]}\Big) \\
    &+   u_{[\mu} b_{\nu]} \left(\Dtu f-\frac{1}{n-1}\Theta f^2\right)
    +\Big[\Dtu b_{[\mu}+fb^\rho \sigma_{\rho[\mu}\Big]\tilde{u}_{\nu]}.\,\nonumber 
    \end{align}
\end{thm}

\section{Physical examples}
\label{sec:Examples}

\subsection{Minkowski spacetime}

Let us briefly discuss observers living in the flat spacetime, that is, $\eta_{\mu\nu}=\text{diag}(-1,1,1,1)$.  Greek indices in this section represent components in the Cartesian frame of the Minkowski spacetime. 
As a starting point, let us consider the simplest tetrad,
\begin{equation}\label{eq:MinktrivialOb}
    e_A{}^\mu= \begin{pmatrix}
    1 & 0 & 0 & 0 \\
    0 & 1 & 0 & 0 \\
    0 & 0 & 1 & 0 \\
    0 & 0 & 0 & 1
    \end{pmatrix}\,,
\end{equation}
providing the observer $u^\mu:=e_{0}{}^\mu=(1,0,0,0)$. It is easy to check that the vorticity tensor $\omega_{\mu\nu}=0$ as well as the other kinematic characteristic are zeros. Therefore, such an observer provides the $3+1$ splitting with a totally geodesic foliation. Moreover, it is easy to see that $h_{\mu\nu}=\eta_{\mu\nu}+u_{\mu}u_{\nu}= \text{diag}(0,1,1,1)$ is an induced metric.

Something interesting happens with this trivial tetrad: Lorentz transformations applied to it can be directly translated into Lorentz transformations of the coordinates.
To see this, consider first an arbitrary tetrad $e_A{}^\mu$ and suppose that we perform a Lorentz transformation $\Lambda^A{}_B$ to obtain another tetrad $\tilde{e}_A{}^\mu$. We can then do the following:
\begin{equation}
    \tilde{e}_A{}^\mu = \Lambda^B{}_A e_B{}^\mu =  \Lambda^B{}_C e_B{}^\mu E^C{}_\nu e_A{}^\nu  = \left(e_B{}^\mu\Lambda^B{}_C  E^C{}_\nu\right)e_A{}^\nu \equiv \xi^\mu{}_\nu  e_A{}^\nu
\end{equation}
where $\xi^\mu{}_\nu := e_B{}^\mu\Lambda^B{}_C  E^C{}_\nu$. For the trivial tetrad \eqref{eq:MinktrivialOb} in Minkowski spacetime, we have that both the components of $e_A{}^\mu$ and $E^A{}_\mu$ are identity matrices. Then, for this particular case, the matrix $\xi^\mu{}_\nu$ is the same Lorentz transformation as $\Lambda^A{}_B$, but acting on the upper index of the tetrad. The observers are then directly related,
\begin{equation}
    \tilde{u}^\mu = \xi^\mu{}_\nu u^\nu\,,
\end{equation}
in contrast to \eqref{eq:LorentzObs}, which also involves the spatial tetrad.

\subsubsection{Boosted observer}\label{sec:boostedobs}

Considering a simple boost
\begin{align}\label{boost1}
    \xi^\mu{}_\nu =\begin{pmatrix}
    \cosh(\lambda(t,x,y,z)) & -\sinh(\lambda(t,x,y,z))& 0 & 0 \\
    -\sinh(\lambda(t,x,y,z)) & \cosh(\lambda(t,x,y,z)) & 0 & 0 \\
    0 & 0 & 1 & 0 \\
    0 & 0 & 0 & 1
    \end{pmatrix}.
\end{align}
of the tetrad discussed above . This operation generates a new tetrad, $\tilde{e}_A{}^\mu= \xi^\mu{}_\nu e_A{}^\nu$, and the boosted observer reads now
\begin{align}
    \tilde{u}^\mu:= \tilde{e}_{0}{}^\mu =(\cosh(\lambda),-\sinh(\lambda),0,0)\,,\label{eq:boostedu}
\end{align}
which corresponds to the previously defined function $f$ and vector $b^\mu$:
\begin{equation}
    f(\lambda)=\cosh(\lambda)\,,\qquad b^\mu(\lambda) = (0,-\sinh(\lambda),0,0)\,.
\end{equation}
For such an observer, the kinematic characteristics are
\begin{align}
    \tilde{\dot{u}}^\mu &= \Big(\cosh(\lambda)\sinh(\lambda) \lambda_t - \sinh^2(\lambda) \lambda_x\,,\ -\cosh^2(\lambda) \lambda_t - \cosh(\lambda)\sinh(\lambda) \lambda_x\,,\ 0,\ 0\Big)\,,\\
    \tilde\Theta&= \sinh(\lambda)\lambda_t - \cosh(\lambda)\lambda_x,\\
    \tilde\omega_{\mu\nu}&=-\frac{1}{2}\begin{pmatrix}
    0 & 0 &  \sinh(\lambda)\lambda_y & \sinh(\lambda)\lambda_z \\[0.5em]
      & 0 &  \cosh(\lambda) \lambda_y&  \cosh(\lambda)\lambda_z\\[0.5em]
    & & 0 & 0\\
    {\rm antis.} & & & 0
    \end{pmatrix}\,,\\
    \tilde\sigma_{\mu\nu}&= \begin{pmatrix}
    \frac{2}{3}\sinh^2(\lambda)\tilde\Theta & \frac{2}{3}\cosh(\lambda)\sinh(\lambda)\tilde\Theta & 
    -\frac{1}{2} \sinh(\lambda)\lambda_y & -\frac{1}{2} \sinh(\lambda) \lambda_z\\[0.5em]
    & \frac{2}{3}\cosh^2(\lambda)\tilde\Theta &-\frac{1}{2}  \cosh(\lambda)\lambda_y & -\frac{1}{2}\cosh(\lambda)\lambda_z\\[0.5em]
    & & -\frac{1}{3}\tilde\Theta & 0\\[0.5em]
    {\rm sym.} & & & -\frac{1}{3}\tilde\Theta
    \end{pmatrix}\,,
\end{align}
where we have defined $\lambda_\mu \equiv \partial_\mu \lambda $. Notice that, since $\cosh(\lambda)>0$, we can conclude:
\begin{prop}\label{result:boosted}
    Let $(\mathcal{M},\eta)$ be the Minkowski spacetime. The boosted observer \eqref{eq:boostedu} defines a foliation if and only if $\lambda_y=\lambda_z=0$. 
\end{prop} 
Therefore, we can still deal with $3+1$ decomposition with the boosted observer, which provides a foliation, however with a distinct geometry. In comparison to the previous case, we are dealing now with a non-trivial embedding, since the second fundamental form (extrinsic curvature) with these coordinates is given by:
\begin{align}\label{boosted_proj}
     K_{\mu\nu}&= \big(\sinh(\lambda)\lambda_t-\cosh(\lambda)\lambda_x\big)\nonumber
    \times
     \begin{pmatrix}
    \sinh^2(\lambda) & \sinh(\lambda)\cosh(\lambda)& 0 & 0 \\
  \sinh(\lambda)\cosh(\lambda) & \cosh^2(\lambda) & 0 & 0 \\
    0 & 0 & 0 & 0 \\
    0 & 0 & 0 & 0
    \end{pmatrix}\,.
\end{align}

Under the condition from Proposition \ref{result:boosted}, the shear simplifies as:
\begin{align}
   \tilde\sigma_{\mu\nu}&= \frac{1}{3}\tilde\Theta\begin{pmatrix}
    2\sinh^2(\lambda) & 2\cosh(\lambda)\sinh(\lambda) & 0 & 0 \\[0.5em]
    2\cosh(\lambda)\sinh(\lambda)& 2\cosh^2(\lambda) & 0 & 0 \\[0.5em]
    0 & 0 & -1 & 0\\[0.5em]
    0 & 0 & 0  & -1
    \end{pmatrix}\,.
\end{align}
We observe that the resulting foliation has non-vanishing expansion and shear. Hence, in general, the solution does not meet the requisites to be minimal, umbilical, or geodesic. The induced metric on the hypersurfaces orthogonal to the observer is:
\begin{align}\label{boosted_proj}
     h_{\mu\nu}=\begin{pmatrix}
    \sinh^2(\lambda) & \sinh(\lambda)\cosh(\lambda)& 0 & 0 \\
  \sinh(\lambda)\cosh(\lambda) & \cosh^2(\lambda) & 0 & 0 \\
    0 & 0 & 1 & 0 \\
    0 & 0 & 0 & 1
    \end{pmatrix}\,.
\end{align}

%%%%%%%%%%%%%%%%%%%%%%%%%%%%

\subsubsection{Observer under Lorentz rotations}

Let us now discuss rotating observers in Minkowski spacetime in more detail. Consider the 4-vector field as
\begin{align}\label{urot}
    \tilde{u}^\mu=\left(\sqrt{2},\frac{-y}{\rho},\frac{x}{\rho},0\right),\qquad \rho\equiv \sqrt{x^2+y^2}\,.
\end{align}
It is straightforward to calculate its acceleration and the rest of the kinematic characteristics:
\begin{align}
    \tilde{\dot{u}}^\mu &= -\rho^{-2} (0, x, y, 0)\,,\\
    \tilde{\Theta} &= 0\,,\\
    \tilde{\omega}_{\mu\nu}&=-\frac{1}{\sqrt{2}\rho^2}\begin{pmatrix}
    0 & x &  y  & 0 \\[0.5em]
      & 0 & \sqrt{2}\rho&  0\\[0.5em]
    & & 0 & 0\\
    {\rm antis.} & & & 0
    \end{pmatrix}\,,\\
    \tilde{\sigma}_{\mu\nu}&=\frac{1}{\sqrt{2}\rho^3}\begin{pmatrix}
    0 & x\rho &  y\rho  & 0 \\[0.5em]
      & 2\sqrt{2} xy &  \sqrt{2}(y^2-x^2)&  0\\[0.5em]
    & & -2\sqrt{2} xy & 0\\
    {\rm sym.} & & & 0
    \end{pmatrix}\,.
\end{align}
Since the vorticity is nonvanishing, this observer does not admit foliation.

The Lorentz transformation connecting any pair of timelike vectors is a boost that can be decomposed into a pure rotation (until the spatial part of the vector is aligned with one of the axis) and a boost with respect to such axis. In our case, $u^\mu=(1,0,0,0)$ and the observer \eqref{urot} can be connected via the composition of a rotation in the $xy$ plane and a boost in the $ty$ direction. The result is: 
\begin{align}\label{boost2}
    \xi^\mu{}_\nu=\frac{1}{\rho}\begin{pmatrix}
    \sqrt{2}\rho & 0 & \rho& 0 \\
   - y & x & -\sqrt{2}y & 0 \\
   x & y & \sqrt{2}x & 0 \\
    0 & 0 &0 & \rho
    \end{pmatrix},
\end{align}
which fulfills $\tilde{u}^\mu = \xi^\mu{}_\nu u^\nu$.

%%%%%%%%%%%%%%%%%%%%%%%%%%%%

\subsection{G\"odel's spacetime}

G\"odel's spacetime is an exact solution to Einstein's field equations with a pressure-free perfect fluid and a negative cosmological constant. The metric has the form (see e.g. \cite{janssen2022}):
\begin{align}
    \dd s^2 &= - \dd t^2 +  \dd r^2 + \dd z^2  + \frac{4}{\Omega}\sinh^2(\Omega r/\sqrt{2})\dd t \dd \varphi \nonumber\\
    &\qquad+ \frac{1}{2\Omega^2}\left[\sinh^2(\sqrt{2}\Omega r)- 8 \sinh^4(\Omega r/\sqrt{2})\right] \dd \varphi^2\,.
\end{align}
The parameter $\Omega$ is both connected with the density of the fluid and the cosmological constant and has to do with an intrinsic angular velocity of the spacetime around the observer at the center of the coordinate system ($r=0$). In fact, in the limit $\Omega\to 0$, we recover Minkowski spacetime in cylindrical coordinates.

It can be checked that the observer defined by the time coordinate, $u^\mu = (1, 0, 0, 0)$, is comoving with the fluid \cite{hawking2023large}. This observer is non-accelerating (i.e., geodesic), expansion free, and shear free. However,  the rotation of the spacetime induces a rotation of the observer that makes it impossible to define an orthogonal foliation (even locally). Indeed, there is a nontrivial component of the vorticity:
\begin{equation}
    \omega_{r\varphi} = -\frac{\sinh(\sqrt{2}\Omega r)}{\sqrt{2}} \,,\\
\end{equation}
which is nonvanishing for nonzero $\Omega$. Moreover, the fact that the spacetime has closed timelike curves implies that it cannot be globally decomposed in a standard 3+1 framework due to the lack of global hyperbolicity.

\subsection{Spherically-symmetric spacetime}

In the following part, we will consider two observers in Schwarzchild spacetime: the accelerating observer
\begin{equation}\label{stat}
 u^\mu=\left(\left(1-\frac{m}{r}\right)^{-\frac{1}{2}},\,0,\,0,\,0\right)\,,
\end{equation}
and the geodesic observer: 
\begin{equation}
 \tilde u^\mu=\left(\left(1-\frac{3m}{r}\right)^{-1/2},\,0,\,\sqrt{\frac{m}{r^2(r-3m)}},\,0\right), \label{eq:Schwutilde}
\end{equation}
where $m$ is a mass of a (compact) object.
For the first one, all the kinematic properties are vanishing except for the acceleration, which has the form:
\begin{equation}
    \dot{u}^\mu = \left( 0,\, \frac{m}{r^2},\, 0,\, 0 \right)\,.
\end{equation}
For the second observer, we find:
\begin{align}
    \tilde{\dot{u}}^\mu &= 0\,,\\
    \tilde \Theta&=\sqrt{\frac{m}{r-3m}}\frac{\cot(\theta)}{r}\,,\\
    \tilde \omega_{\mu\nu} &= -\frac{r-6m}{4}\sqrt{\frac{r}{(r-3m)^3}}
      \begin{pmatrix}
        0 & \frac{m}{r^2 }  & 0 & 0 \\
        & 0 & \sqrt{\frac{m}{r}} & 0 \\
        &  & 0 & 0 \\
        {\rm antis.} &  & & 0 
        \end{pmatrix}\,,\\
   \tilde \sigma_{\mu\nu} & =
    -\frac{m(r-2m)}{3r(r-3m)} \tilde{\Theta} 
        \begin{pmatrix}
        1 &  -\frac{9}{4}\sqrt{\frac{r}{m}} \tan(\theta) & - r \sqrt{\frac{r}{m}}  & 0\\
        & \frac{r^2(r-3m)}{m(r-2m)^2} & \frac{9r^2}{4m} \tan(\theta)& 0 \\
        &  & \frac{r^3}{m} & 0 \\
        {\rm sym.} &  & & -\frac{2r^3(r-3m)}{m (r-2m)}\sin^2(\theta) 
        \end{pmatrix}\,.
\end{align}
We notice that the expansion becomes singular at the poles of the sphere $\theta=0,\pi$.

The vorticity and shear scalars are 
\begin{align}
\omega&:=\sqrt{\omega_{\mu\nu}\omega^{\mu\nu}}=\frac{1}{2}\sqrt{\frac{m}{2r^3}}\left(\frac{r-6m}{r-3m}\right),\\
\sigma&:=\sqrt{\sigma_{\mu\nu}\sigma^{\mu\nu}}\ = \frac{1}{2}\sqrt{\frac{m}{6r^3}} \frac{\sqrt{27(r-2m)^2+16r(r-3m)\cot^2(\theta)}}{r-3m}.
\end{align}

Now we construct a tetrad out of the observer \eqref{stat}:
\begin{equation}
     e_A{}^\mu=\begin{pmatrix}
          \frac{1}{\sqrt{1-\frac{2m}{r}}}  & 0 &  0  & 0 \\
   0 & \sqrt{1-\frac{2m}{r}} & 0 & 0 \\
    0 & 0 &   \frac{1}{r}  & 0 \\
   0 & 0 & 0 & \frac{1}{r\sin(\theta)}
        \end{pmatrix}
\end{equation}
It can be proven that the following Lorentz transformation generates another tetrad with observer \eqref{eq:Schwutilde}:
\begin{equation}
    \Lambda^A{}_B=\begin{pmatrix}
          \sqrt{ \frac{r-2m}{r-3m}}  & 0 &    \sqrt{ \frac{m}{r-3m}}  & 0 \\
   0 & 1 & 0 & 0 \\
    \sqrt{ \frac{m}{r-3m}} & 0 &   \sqrt{ \frac{r-2m}{r-3m}}  & 0 \\
   0 & 0 & 0 & 1
        \end{pmatrix}\,.
\end{equation}

\subsection{Examples in theories beyond General Relativity}

Here we present further examples that correspond to solutions of certain modified gravity theories. In particular, the ones we analyze were taken from the works \cite{ferraro2015remnant, nashed2015regularization} on $f(T)$ gravity\footnote{For readers not familiar with this proposal, see \cite{Golovnev:2018wbh} and references therein.
}.

\subsubsection{Example 1}

Firstly, let us consider the local Lorenz tranformation which was studied in \cite{ferraro2015remnant} applied to the trivial observer \eqref{eq:MinktrivialOb} in Minkowski spacetime:
\begin{equation}
\Lambda^A{}_B ~=~
\begin{pmatrix}
\cosh (\chi)  & 0 & 0 & \sinh (\chi)  \\
0 & \cos (\alpha)  & -\sin( \alpha)  & 0 \\
0 & \sin (\alpha)  & \cos( \alpha)  & 0 \\
\sinh (\chi)  & 0 & 0 & \cosh (\chi)
\end{pmatrix} ~,
\end{equation}
where $\alpha$ and $\chi$ are two independent local parameters $\chi(x^{\mu}) $ and $\alpha( x^{\mu})$ of the transformation. This matrix can be written as the composition $\Lambda^A{}_B=B^A{}_C R^C{}_B$ of a local rotation $R$ with angle $\alpha$ and a local boost $B$ with rapidity $-\chi$:
\begin{equation}
    R^A{}_B=\begin{pmatrix}
        1  & 0 & 0 & 0  \\
        0 & \cos (\alpha)  & -\sin( \alpha)  & 0 \\
        0 & \sin (\alpha)  & \cos (\alpha)  & 0 \\
        0  & 0 & 0 & 1
    \end{pmatrix}\,,\qquad
    B^A{}_B=\begin{pmatrix}
        \cosh (\chi)  & 0 & 0 & \sinh (\chi)  \\
        0 & 1  & 0  & 0 \\
        0 & 0  & 1  & 0 \\
        \sinh (\chi) & 0 & 0 & \cosh (\chi)
    \end{pmatrix}\,.
\end{equation}
These operators commute. So we can analyze them separately. As we know, the rotation preserves the foliation since the observer remains unchanged, so $\alpha$ can be arbitrary. On the other hand, as shown in Sec.~\ref{sec:boostedobs}, for the transformed observer to be foliation-generating, the boost rapidity is required to fulfill $\partial_x \chi = 0 = \partial_y \chi$, which is the same as the condition of solving the equations of motions of $f(T)$ gravity in the Weitzenböck gauge.

\subsubsection{Example 2}

Now we focus on the (inverse) tetrads discussed in \cite{nashed2015regularization}:
\begin{equation} \label{basic}
E^A{}_\mu =
\begin{pmatrix}
A(r) & A_1(r) & 0 & 0 \\
A_2(r)\sin(\theta) \cos(\phi) & A_3(r)\sin(\theta) \cos(\phi) & r\cos(\theta) \cos(\phi) & -r\sin(\theta) \sin(\phi) \\
A_2(r)\sin(\theta) \sin(\phi) & A_3(r)\sin(\theta) \sin(\phi) & r\cos(\theta) \sin(\phi) & r\sin(\theta) \cos(\phi) \\
A_2(r)\cos(\theta) & A_3(r)\cos(\theta) & -r\sin(\theta) & 0 \\
\end{pmatrix}\,,
\end{equation}

\begin{equation}\label{nobasic}
\tilde{E}^A{}_\mu = {\footnotesize
\begin{pmatrix}
LA + HA_2 & LA_1 + HA_3 & 0 & 0 \\
-(LA_2 + HA) \sin(\theta) \cos(\phi) & -(LA_3 + HA_1) \sin(\theta) \cos(\phi) & -r \cos(\theta) \cos(\phi) & r \sin(\theta) \sin(\phi) \\
-(LA_2 + HA) \sin(\theta) \sin(\phi) & -(LA_3 + HA_1) \sin(\theta) \sin(\phi) & -r \cos(\theta) \sin(\phi) & -r \sin(\theta) \cos(\phi) \\
-(LA_2 + HA) \cos(\theta)            & -(LA_3 + HA_1) \cos(\theta) & r \sin(\theta) & 0
\end{pmatrix}}
\end{equation}
where $A(r)$, $A_1(r)$, $A_2(r)$, and  $A_3(r)$ are four unknown functions of the radial coordinate , $L=L(r):=\sqrt{H(r)^2+1}$ and $H=H(r)$ is an arbitrary function.
 
These tetrad fields have been studied in \cite{nashed2013special} and, in particular, the case with
\begin{align}\label{Sch_ba}
A(r)=1-\dfrac{M}{r}, \qquad A_1(r)=\dfrac{M}{r}\left(1-\frac{M}{r}\right)^{-1},\quad
A_2(r)=\dfrac{M}{r}, \qquad
A_3(r)=\dfrac{1-\dfrac{M}{r}}{1-\dfrac{2M}{r}}\,,
\end{align}
which corresponds to a solution of the theory,  provided that $ f(0)=0, f_T(0)\neq0,  f_{TT}\neq0$.
 
The local Lorentz transformation connecting the inverse tetrads,
\begin{equation}
    \tilde{E}^A{}_\mu =\Lambda^A{}_B E^B{}_\mu,
\end{equation}
has the following form\footnote{{Here we are correcting some sign typos present in \cite{nashed2013special}.}}
\begin{equation}
\Lambda^A{}_B =  
-{\footnotesize
\begin{pmatrix}
-L & -H \sin(\theta) \cos(\phi) & -H \sin(\theta) \sin(\phi) & -H \cos(\theta) \\
 H \sin(\theta) \cos(\phi) & 1 + H_1 \sin^2(\theta) \cos^2(\phi) & H_1 \sin^2(\theta) \sin(\phi) \cos(\phi) & H_1 \sin(\theta) \cos(\theta) \cos(\phi) \\
H \sin(\theta) \sin(\phi) & H_1 \sin^2(\theta) \sin(\phi) \cos(\phi) & 1 + H_1 \sin^2(\theta) \sin^2(\phi) & H_1 \sin(\theta) \cos(\theta) \sin(\phi) \\
H \cos(\theta) & H_1 \sin(\theta) \cos(\theta) \cos(\phi) & H_1 \sin(\theta) \cos(\theta) \sin(\phi) & 1 + H_1 \cos^2\theta \\
\end{pmatrix}}
\end{equation}
where $H_1 = L-1$.
It can be checked that all the observers associated to the transformed tetrad are given by the following expression:
\begin{equation}
\tilde{u}^\mu=\tilde{e}_0{}^\mu = \left(-\frac{A_1 H +A_3 L}{A_1A_2-A_3A},\frac{A H +A_2 L}{A_1A_2-A_3A},0,0\right)\,.
\end{equation}
All these observers are vorticity-free and have associated local foliations.

\subsubsection{Example 3}
The tetrad presented in \cite{PhysRevD.99.024022} coincides with the one proposed earlier in the literature \cite{Tamanini:2012hg,lucas2009regularizing,krvsvsak2016covariant}
in the case of the flat limit but can provide different solutions in the strong field regime. Using the non-flat metric in spherical coordinates
\begin{align}
    \mathrm{d}s^2=-A(r)\mathrm{d}t^2+B(r)\mathrm{d}r^2+r^2\mathrm{d}\Omega^2\,
\end{align} 
the authors consider an inverse tetrad obtained after applying a boost to the diagonal one, whose associated observer is
\begin{equation}
    u^\mu = \left(\frac{1}{\sqrt{A}},0,0,0\right)\,.
\end{equation}
The result is:
\begin{align}
    \tilde{E}^A{}_{\mu}=\begin{pmatrix}
        1 & -\sqrt{1-A} & 0 & 0 \\
        -\sqrt{\frac{B}{A}}\sqrt{1-A} & \sqrt{\frac{B}{A}} & 0 &0 \\
        0 & 0 & r & 0\\
        0 & 0 & 0 & r \sin(\theta)
    \end{pmatrix}\,,
\end{align}
whose dual tetrad is given by:
\begin{align}
    \tilde{e}_A{}^{\mu}=
    \begin{pmatrix}
        \frac{1}{A} & \sqrt{\frac{1-A}{AB}} & 0 & 0 \\
        \frac{\sqrt{1-A}}{A} & \frac{1}{\sqrt{AB}} & 0 &0  \\
        0 & 0 & \frac{1}{r} & 0\\
        0 & 0 & 0 & \frac{1}{r\sin(\theta)}
    \end{pmatrix}\,.
\end{align}
We study the first element of this frame, which corresponds to the observer:
\begin{equation}
    \tilde{u}^\mu = \tilde{e}_0{}^\mu = \left(\frac{1}{A}, \sqrt{\frac{1-A}{AB}}, 0, 0\right)\,.
\end{equation}
This observer has a non-vanishing expansion and shear, but it is free falling (zero acceleration) and vorticity-free. So it defines a foliation orthogonal to it. Note that in \cite{PhysRevD.99.024022} there was a pure Lorentz rotation in addition to the boost in order to satisfy the EoM of $f(T)$. But according to Proposition \ref{prop_rotation} this detail will not affect foliation.

\section{Conclusions}\label{sec:Conclusions}

At the beginning of this article, the geometrical notion of foliation has been presented, as well as the conditions that a congruence of observers must fulfill to be orthogonal to one. This is done through known mathematical results in a pedagogical way, providing a straightforward toolkit for physicists to follow. In particular, the kinetic characteristics of the observer's motion include the so-called vorticity (see Definition \ref{defkinem}). Such a tensor is crucial for the matter of the present work since vanishing vorticity is a necessary and sufficient condition for the existence of foliation (see Theorem \ref{th:condit}). This is of particular importance in gravitational physics when studying the Hamiltonian evolution of certain initial conditions for the spacetime (together with the matter fields contained in it). In this context, the $3+1$ split of spacetime is based on the assumption that it admits a (global) foliation.

In the case we rely on a certain observer to construct or determine such a foliation, some care must be taken. If the observer is poorly chosen, such that it does not admit foliation, it can lead to inconsistencies in physical applications such as numerical relativity and Hamiltonian analysis.

In the context of modified gravity theories, it is usual to work in terms of (orthonormal) tetrads, whose timelike vector can be taken as the reference observer to construct the would-be foliation. In some of these theories, there are Lorentz transformations connecting different solutions of the theory and, hence, distinct observers allowed by the dynamical equations. Regarding this point, we have revised how Lorentz transformations affect the foliation condition in a theory-independent way (at the kinematical level). In particular, we found that pure Lorentz rotations acting on an observer that admits foliation will also admit foliation with an unchanged geometry (Proposition \ref{prop_rotation}). Lorentz boosts, on the other hand, can generate non-trivial vorticity for the new observer, which therefore will not admit foliation. However, we also note that boost does not necessarily introduce vorticity, but certain boosts may introduce non-vanishing shear, for instance, while preserving foliation. The general equation that the new observer must satisfy to generate a foliation is presented in Theorem \ref{th:cond_foliation}. To illustrate these facts, we provide examples in Minkowski and Schwarzschild spacetimes in which the transformation leads to a violation of the foliation condition. Note that those are solutions of General Relativity, whose tetrad formulation is Lorentz invariant. This means that one can always locally transform the solution to one that coincides with an observer which admits foliation. This fact may not hold in theories in which Lorentz invariance is broken. As an example, we have investigated known solutions to $f(T)$ gravity in the Weitzenböck gauge. In principle, inconsistencies can arise if one adopts solutions that do not admit foliation. In particular, we presented several examples of observers with vanishing vorticity.

At this point, we would like to comment on some aspects present in the literature that might lead to misinterpretations of the mathematics behind a proper $3+1$ split. In the context of metric-affine theories, one can write the conditions we derived in an equivalent way by invoking a general connection \cite{Agashe2023, Capozziello:2021pcg}. Although this can be interesting from a mathematical point of view, it is a physically misleading way of presenting them and definitely not practical. For instance, such expressions might suggest that the foliation condition (by spacelike hypersurfaces) depends on the chosen connection, whereas it is a purely topological and metric issue.

To conclude, we would like to point out some general thoughts on the potential limitations and applications of the present work. With the increasing interest in modified gravity, it will eventually be required to test it through numerical techniques, as the ones employed in numerical relativity. Given a spacetime solution to a theory, together with an identification of an observer, it is straightforward to check whether or not it satisfies the condition for foliation. However, to develop a full setup for numerical relativity,  much more work is required, including setting up a $3+1$ split that admits the condition for strong hyperbolicity. We finally highlight that the setting employed in this work gives a more practical approach to developing numerical techniques in metric-affine theories of gravity, where the simple Levi-Civita connection is used instead of defining new quantities related to other connections.

%%%%%%%%%%%%%%%%%%%%%%%%%%%%%%%%%%%%%%%%%%
\vspace{6pt}

%%%%%%%%%%%%%%%%%%%%%%%%%%%%%%%%%%%%%%%%%%
\authorcontributions{

Conceptualization, D.B., A.J.C, A.W.; 
methodology, D.B., A.J.C, A.W.; 
software, A.J.C.; 
validation, D.B., A.J.C; 
formal analysis, D.B., A.J.C, A.W.; 
investigation, D.B., A.J.C, A.W.; 
resources, D.B., A.J.C, A.W.; 
%data curation, X.X.; 
writing---original draft preparation, D.B., A.J.C, A.W.;
writing---review and editing, D.B., A.J.C, A.W.;
%visualization, X.X.; 
supervision, A.W.; 
project administration, A.W.; 
funding acquisition, D.B., A.J.C, A.W. 

All authors have read and agreed to the published version of the manuscript.
% please turn to the  \href{http://img.mdpi.org/data/contributor-role-instruction.pdf}{CRediT taxonomy} for the term explanation. Authorship must be limited to those who have contributed substantially to the work~reported.
}

\funding{

A.J.C. was supported by the Estonian Research Council through the grant PRG356.  A.W. acknowledges financial support from MICINN (Spain) {\it Ayuda Juan de la Cierva - incorporaci\'on} 2020 No. IJC2020-044751-I and from  the Spanish Agencia Estatal de Investigaci\'on Grant No. PID2022-138607NB-I00, funded by MCIN/AEI/10.13039/501100011033, EU and
ERDF A way of making Europe.}

\dataavailability{No data has been generated during the present study.} 

\acknowledgments{
The authors would like to thank Mar\'ia-Jos\'e Guzm\'an, Margus Saal and Gerardo Garc\'ia-Moreno for useful comments and discussions.}

\conflictsofinterest{
The authors declare no conflicts of interest.} 

%%%%%%%%%%%%%%%%%%%%%%%%%%%%%%%%%%%%%%%%%%
\appendixtitles{yes}
\appendixstart
\appendix
\section{Geometry of distributions and foliations}\label{app:distrib}

Let us introduce \cite{gray1967pseudo,yano1985structures} the notion of an almost-product structure on $\mathcal{M}$ and further Naveira's classification of the almost-product manifolds \cite{naveira1983classification}. The structure is determined by 
a field of endomorphisms of $T\mathcal{M}$, i.e., a $(1,1)-$tensor
field $P$ on $\mathcal{M}$, such that $P^2=\idop$ ($\idop$ is the identity operator). In this case, at any point $p\in \mathcal{M}$, one can consider two 
subspaces of $T_p\mathcal{M}$ corresponding respectively to two eigenvalues $\pm 1$ of $P$.
It defines two complementary  distributions on $\mathcal{M}$, i.e., $T\mathcal{M}=\mathcal{D}^+\oplus \mathcal{D}^- $.
Moreover, if $\mathcal{M}$ is equipped with a pseudo-Riemanian metric $g$ such that:
\begin{equation}\label{comp}
 g(PX,PY)=g(X,Y),\;\;X,Y\in\Gamma(T\mathcal{M}),
\end{equation}
then both distributions are mutually orthogonal. In this case, $P$ is called a pseudo-Riemaniann almost-product structure.

Let now $\mathcal{D}$ be a $k$-distribution on $(\mathcal{M},g)$ and $\mathcal{D}^{\perp}$ the distribution orthogonal to $\mathcal{D}$. At every 
point $p\in \mathcal{M}$, we have then $T_pM=\mathcal{D}_p\oplus \mathcal{D}_p^{\perp}$.
Thus we can uniquely define a $(1,1)$ tensor field $P$ such that $P^2=\idop_{T\mathcal{M}},\; P|_{\mathcal{D}}=\idop_\mathcal{D},\; P|_{\mathcal{D}^{\perp}}=-\idop_{\mathcal{D}^{\perp}}$. It is clear
that $P$ becomes automatically a (pseudo-)Riemaniann almost-product structure.

In the following, we follow \cite{Ferrando_2007}. We call $\Pi$ and $\Pi_\perp$ the projectors of vector fields onto $\mathcal{D}$ and $\mathcal{D}^\perp$, respectively. Notice that the following relations hold:
\begin{equation}
    \idop =\Pi+\Pi_\perp\,,\qquad P = \Pi - \Pi_\perp\,.
\end{equation}
We define the vector valued 2-covariant tensor $\mathcal{Q}$ as: 
\begin{equation}
    \mathcal{Q}(X,Y) := \Pi_\perp \Big(\nabla_{\Pi X} (\Pi Y)\Big)\,,
\end{equation}
defined on arbitrary vector fields $X,Y\in\Gamma(T\mathcal{M})$. In index notation:
\begin{equation}
    \mathcal{Q}^\rho{}_{\mu\nu} X^\mu Y^\nu := \Pi_\perp^\rho{}_\alpha (\Pi X)^\nu \nabla_\beta (\Pi Y)^\mu\,,
\end{equation}
where $(\Pi X)^\mu := \Pi^\mu{}_\nu X^\nu$\,. This object can be split as follows:
\begin{align}
    \mathcal{Q}^\rho{}_{\mu\nu} &= 
    \mathcal{A}^\rho{}_{\mu\nu} + \mathcal{S}^\rho{}_{\mu\nu}\\
    &= 
    \mathcal{A}^\rho{}_{\mu\nu} + \hat{\mathcal{S}}^\rho{}_{\mu\nu} + \frac{1}{k}\mathcal{T}^\rho h_{\mu\nu}
\end{align}
where
\begin{align}
    \mathcal{A}^\rho{}_{\mu\nu} &= \mathcal{Q}^\rho{}_{[\mu\nu]}\,,\\
    \mathcal{S}^\rho{}_{\mu\nu} &= \mathcal{Q}^\rho{}_{(\mu\nu)}\,\\
    \mathcal{T}^\rho &= \mathcal{S}^\rho{}_{\mu\nu} g^{\mu\nu}\,.
\end{align}
$\hat{\mathcal{S}}^\rho{}_{\mu\nu}$ is then implicitly defined as the traceless part of $\mathcal{S}^\rho{}_{\mu\nu}$ in the last two indices and $h_{\mu\nu}:= g_{\alpha\beta} \Pi^\alpha{}_\mu \Pi^\beta{}_\nu$ is the induced metric in $\mathcal{D}$.

\begin{defi}
We say that $\mathcal{D}$ is:
\begin{itemize}
    \item a foliation if $\mathcal{A}^\rho{}_{\mu\nu}=0$;
    \item  geodesic if $\mathcal{S}^\rho{}_{\mu\nu}=0$;
    \item minimal if $\mathcal{T}^\rho=0$;
    \item umbilical if $\hat{\mathcal{S}}^\rho{}_{\mu\nu}=0$.
\end{itemize}    
\end{defi}

Now we present a characterization of these properties in terms of $P$  (see \cite{gil1983geometric}):\footnote{In index notation and arbitrary coordinates $\{x^\mu\}$, we can write:
\begin{align}
    A &= A^\mu \partial_\mu\,,&
    B &= B^\mu \partial_\mu\,,&
    X &= X^\mu \partial_\mu\,,\nonumber\\
    \alpha &= \alpha_\mu {\rm d}x^\mu\,,&
    e_a&=(e_a)^\mu\partial_\mu\,, &
    P &= P^\mu{}_\nu \partial_\mu\otimes {\rm d}x^\nu\,,
\end{align}
and introduce $P_{\mu\nu} \equiv g_{\mu\rho}P^\rho{}_\nu$,
so that the properties $D_1$, $D_2$ and $D_3$ can be expressed, respectively, as
\[
    0=A^\nu A^\rho\nabla_\rho P^\mu{}_\nu \,,\qquad
    0=\alpha_\mu X^\mu \,,\qquad
    0= X^\rho A^\mu B^\nu \left(\nabla_{(\mu|} P_{\rho|\nu)}  - \frac{1}{k} g_{\mu\nu} \alpha_\rho \right)\,,
\]
and $\alpha_\mu= \delta^{ab} (e_a)^\rho (e_b)^\lambda \nabla_\rho P_{\mu\lambda}$, where $\delta^{ab}\equiv {\rm diag}(1,1,\ldots,1)$.
}
\begin{thm}
    The distribution $\mathcal{D}$  is called: geodesic, minimal or umbilical if and only if $\mathcal{D}$ has 
property $D_{1}$, $D_{2}$, $D_{3}$, respectively, where:
\begin{itemize}
\item $D_{1}\Longleftrightarrow\;(\nabla_{A}P)A=0$,

\item $D_{2}\Longleftrightarrow\;\alpha(X)=0$,

\item $D_{3}\Longleftrightarrow\; g((\nabla_{A}P)B,X)+g((\nabla_{B}P)A,X)=\frac{2}{k}g(A,B)\alpha(X)$,
\end{itemize}
where $X\in \Gamma(\mathcal{D})^{\perp},\; A,B\in \Gamma(\mathcal{D}).$
Here $\{e_{a}\}_{a=1}^{k},\;$$(k=\mathrm{dim}\mathcal{D})$ is a local orthonormal frame of $\mathcal{D}$
and $\alpha(X)={\displaystyle \sum_{a=1}^{k}g\big((\nabla_{e_{a}}P)e_{a},X\big)}$.
\end{thm}

It implies that a distribution has the property $D_1$ if and only if it has the properties $D_2$ and $D_3$. Note that the covariant derivatives are Levi-Civita of $g$, that is, it is the metric the almost product structure is compatible with. We then have the following theorem:
\begin{thm}{\bf (Gil-Medrano)}\\
A foliation $\mathcal{D}$ is a totally geodesic, minimal or totally umbilical if and only if $\mathcal{D}$ has the 
property $F_1$, $F_2$, $F_3$ respectively, where
\begin{equation}
 F_i\,\iff\,F+D_i,\;\;i=1,2,3
\end{equation}
and
\begin{equation}\label{calkow}
 F\,\iff\,(\nabla_A P)B=(\nabla_B P)A\;\;\forall\,A,B\in\Gamma(\mathcal{D}).
\end{equation}
which can be written in index notation as
\begin{equation}
    \nabla_\mu P^\rho{}_\nu (A^\mu B^\nu - A^\nu B^\mu) = 0\,.
\end{equation}
\end{thm}

The proof of this theorem can be found in \cite{gil1983geometric}.
It is easy to see that the property $F$ is equivalent to Frobenius' theorem, i.e., the distribution $\mathcal{D}$ with this property
is a maximal foliation. The meanings of totally geodesic, minimal or totally umbilical foliations are briefly discussed in \cite{Borowiec:2013kgx}.

Coming back to our physical examples discussed in this manuscript, it is straightforward to see that the tensor constructed with the given observer $u^\alpha$ and spacetime metric $g_{\alpha\beta}$
\begin{equation}
    P^\alpha{}_{\beta}= \delta^\alpha{}_{\beta} + 2u^\alpha u_\beta
\end{equation}
satisfies the properties of the almost-product structure. Therefore, since $P$ is compatible with the metric $g$, the results \cite{awmgr,Borowiec:2013kgx} on geometry and its relation to the observer's kinematics are given with respect to $\nabla$ being compatible with $g$.

%%%%%%%%%%%%%%%%%%%%%%%%%%%%%%%%%%%%%%%%%%
\begin{adjustwidth}{-\extralength}{0cm}

\reftitle{References}

\bibliography{references.bib}

\PublishersNote{}
\end{adjustwidth}
\end{document}